\documentclass[reqno]{amsart}
\usepackage{amsfonts}
\usepackage{amssymb}
\usepackage[dvips]{graphics}
\usepackage{epsfig}
\usepackage{color}

\usepackage{hyperref}

\usepackage{fancyvrb}

 \newtheorem{thm}{Theorem}[section]
  \newtheorem*{thm*}{Theorem}

 \theoremstyle{definition}

 \numberwithin{equation}{section}

\newcommand{\caA}{{\mathcal A}}

\newcommand{\caF}{{\mathcal F}}

\newcommand{\caH}{{\mathcal H}}

\newcommand{\caO}{{\mathcal O}}

\newcommand{\bbC}{{\mathbb C}}

\newcommand{\bbN}{{\mathbb N}}

\newcommand{\bbR}{{\mathbb R}}

\newcommand{\bbZ}{{\mathbb Z}}

\newcommand{\iu}{\mathrm{i}}

\newcommand{\str}{^{*}}
\newcommand{\ep}[1]{\mathrm{e}^{#1}}

\newcommand{\dd}{\,\mathrm{d}}
\newcommand{\tr}{\mathrm{tr}}

\newcommand{\be}{\begin{equation}}
\newcommand{\ee}{\end{equation}}
\newcommand{\bea}{\begin{eqnarray}}
\newcommand{\eea}{\end{eqnarray}}
\newcommand{\beann}{\begin{eqnarray*}}
\newcommand{\eeann}{\end{eqnarray*}}

\newcommand{\dist}{\mathrm{dist}}

\newcommand{\eqL}{\stackrel{\scriptscriptstyle L}{=}}

\newcommand{\ring}{\Pi}

\title{On the absence of stationary currents}

\author{Sven Bachmann}
\address{Department of Mathematics \\ The University of British Columbia \\ Vancouver, BC V6T 1Z2 \\ Canada}
\email{sbach@math.ubc.ca}

\author{Martin Fraas}
\address{Department of Mathematics \\ Virginia Tech \\ Blacksburg, VA 24061-0123 \\ USA}
\email{fraas@vt.edu}

\date{\today}

\begin{document}

\begin{abstract}
We review proofs of a theorem of Bloch on the absence of macroscopic stationary currents in quantum systems. The standard proof shows that the current in 1D vanishes in the large volume limit under rather general conditions. In higher dimension, the total current across a cross-section does not need to vanish in gapless systems but it does vanish in gapped systems. We focus on the latter claim and give a self-contained proof motivated by a recently introduced index for many-body charge transport in quantum lattice systems having a conserved $U(1)$-charge.
\end{abstract}

\maketitle
 

Although he never published the result himself, Bloch proved in the early thirties that the state of lowest free energy cannot carry a net mean current. The first published proof thereof is due to Bohm~\cite{Bohm}, who further argues for the vanishing of the current in a ring of large radius. More recently, Ohashi and Momoi revisited the proof in a second quantized setting~\cite{Ohashi}, while Yamamoto~\cite{yamamoto} and Watanabe~\cite{watanabeBloch} insisted on the role of gauge invariance and the relation to the Lieb-Schultz-Mattis theorem~\cite{LSM2}. Finally, the many-body index introduced in~\cite{MBIndex} provides an alternative point of view on the proof of Bloch's theorem, which extends to higher dimensional settings but requires a gap.

This note focuses on interacting lattice systems having a $U(1)$-conserved charge, defined on a ring or a higher dimensional torus. While widely known, the theorem has a somewhat vague status in the mathematics literature: the goal of this note is to clarify both the statements and the assumptions. In particular, we note that time reversal symmetry is not required for the validity of the statement. 

While we shall give precise versions of the theorem later, we immediately provide an informal statement.
\begin{thm*}
For a gapped system, the ground state expectation value of the total current flowing through a hypersurface vanishes. In 1D, the same holds without the gap assumption.
\end{thm*}
For non-interacting electrons on a line, the theorem is a problem in spectral theory. Let $H$ be the single-particle finite range Hamiltonian. Then  the (single-particle) current $J$ across a fiducial point is 
\begin{equation}\label{First J}
J = \iu[H,Q]
\end{equation}
with $Q$ the charge on the half-line originating at the point. Since $J$ is a trace class operator, the many-body ground state expectation is
\begin{equation*}
\tr(P_F \iu[H,Q]) = \iu \, \tr(P_F [H, \overline Q])
\end{equation*}
where $P_F$ is the Fermi projection, namely the projection onto the eigenspace of $H$ corresponding to energies below the Fermi energy $E_F$, and $\overline{Q} = P_F Q P_F$. If the Fermi energy lies in a gap of $H$, then $\overline Q$ is a projection up to a trace-class perturbation, see~\cite{MBIndex}, and it has a complete set of eigenvectors. In this basis,  all diagonal elements of $P_F [H, \overline Q]$ are zero, and the trace vanishes by Lidskii's theorem.

The many-body charge on a half-line is in general not a well-defined operator. In fact, the description of interacting systems in the infinite volume limit is at best technically involved. As is customary, we shall rather consider a system confined to a box of width $L$ (here with periodic boundary conditions), compute the expectation value of the current, and only then send $L$ to infinity. This finite volume approach has the advantage of providing explicit bounds for convergence rates, which, as we shall see, carry physical information.

Crucially, the current in a finite volume cannot be expressed as a rate of change of charge as in~(\ref{First J}). For a charge operator $Q$ on a segment with two endpoints, the commutator $\iu[H,Q]$ is a sum of two contributions supported in a neighbourhood of these endpoints. The relevant current operator $J$ corresponds to only one of these two contributions, and we show that each one of these two contributions vanishes separately. The above purely spectral theoretic argument cannot be applied and we need techniques to restrict the attention to subregions of the physical space. 

From this point of view, Bloch's theorem is one of the simplest examples of the use of spectral theory in a many-body setting with locality constraints. The toolkit that we are going to use was developed in \cite{LSM, HastingsLesHouches, Sven}.

We will present two arguments for the proof of the theorem. The first one goes back to the original article of Bohm and was sharpened considerably in later works, see in particular~\cite{watanabeBloch}. It shows that the current vanishes as $L^{-1}$ in 1D . The second argument is parallel to that of our recent work~\cite{MBIndex}. It requires a gap condition but shows that the total current vanishes as $L^{-\infty}$
 in any dimension.

\section{The classical proofs: a variational argument}\label{sec:classical}

In this section, we follow \cite{watanabeBloch}. We consider interacting electrons on a discrete ring $\ring=\bbZ/L\bbZ$. The charge operator at site $x$ is given by $q_x = a_x\str a_x$, the charge in an interval is $Q_{[a,b]} = \sum_{x=a}^b q_x$ and the total charge is $Q_\ring = \sum_{x\in \ring} q_x$.

The two essential assumptions on the system's Hamiltonian $H$ are first its {locality} $H = \sum_{x\in\ring} h_{x,x+1}$, and second the {charge conservation}
\begin{equation*}
[H, Q_\ring] = 0.
\end{equation*}
Of course, the nearest neighbour assumption is there for simplicity, the argument below is easily adapted to any finite range Hamiltonian. Together, the two assumptions imply that for any $x\in \ring$, $h_{x,x+1}$ can be chosen to individually satisfy 
\begin{equation}\label{1d Charge conservation}
[h_{x,x+1},Q_\ring] = 0.
\end{equation}
Indeed, suppose that $H = \sum_{x\in \Pi}\tilde h_{x,x+1}$, then by charge conservation $H = \sum_{x\in \ring} h_{x,x+1}$, where
\begin{equation*}
 h_{x,x+1} = \frac{1}{2\pi}\int_0^{2\pi}\ep{\iu \theta Q_\ring} \tilde h_{x,x+1}\ep{-\iu \theta Q_\ring}d\theta
\end{equation*}
is still supported on $\{x,x+1\}$ and for which
\begin{equation*}
[ h_{x,x+1},Q_\ring] = \frac{\iu}{2\pi}\left.\ep{\iu \theta Q_\ring} \tilde h_{x,x+1}\ep{-\iu \theta Q_\ring}\right\vert_{\theta=0}^{\theta=2\pi} = 0
\end{equation*}
since $Q_\ring$ has integer spectrum.

For later use, we write the term $h_{x,x+1}$ as a polynomial, $p_{x,x+1}$, in the local creational and annihilation operators, $h_{x,x+1} = p_{x,x+1}(a_x, a^*_x, a_{x+1}, a^*_{x+1})$. Condition (\ref{1d Charge conservation}) is equivalent to 
\begin{equation}
\label{1.2}
p_{x,x+1}(\ep{\iu \varphi} a_x, \ep{-\iu \varphi}a^*_x, \ep{\iu\varphi} a_{x+1}, \ep{-\iu \varphi} a^*_{x+1}) = p_{x,x+1}(a_x, a^*_x, a_{x+1}, a^*_{x+1}).
\end{equation}
The polynomial can depend on $x$, but all coefficients are assumed to be bounded by a constant $C$. Note that here and below, such constants are always independent of~$L$.

Local gauge transformations are determined by a function $\theta: \ring\to\bbR$, and implemented by the corresponding unitary
\begin{equation*}
U_\theta = \ep{\iu ( \theta,q)},\qquad ( \theta,q)=\sum_{x\in \ring}\theta_x q_x.
\end{equation*}
The gauge transformed $H_\theta = U_\theta\str H U_\theta$ satisfies the relation
\begin{equation}
\label{in_between}
\left(f,\nabla H_\theta\right) = \sum_{x\in \ring} f_x \iu [H_\theta,q_x],
\end{equation}
for any function $f: \ring\to\bbR$. We can write $H_\theta$ explicitly as
\begin{equation}
\label{eq:Hta}
H_\theta = \sum_{x\in \ring} p_{x,x+1} (\ep{\iu\theta_x}a_x, \ep{-\iu\theta_x}a_x\str, \ep{\iu\theta_{x+1}}a_{x+1}, \ep{-\iu\theta_{x+1}}a_{x+1}\str).
\end{equation}

By locality and charge conservation, the operator $\iu[H,Q_{[a,b]}]$ is the difference of two currents, one along the edge $\langle a-1,a\rangle$, the other one along $\langle b,b+1\rangle$:
\begin{equation}\label{1d currents}
\iu[H,Q_{[a,b]}] = J_{\langle a-1,a\rangle} - J_{\langle b,b+1\rangle}.
\end{equation}
Using (\ref{in_between}), we have $J_{\langle a-1,a \rangle} - J_{\langle b,b+1\rangle} = (\chi_{[a,b]},\nabla H_\theta)\vert_{\theta=0} = \partial_s H_{s \chi_{[a,b]}}\vert_{s=0}$, where $\chi_{[a,b]}$ is the characteristic function of the interval $[a,b]$. Comparing with (\ref{eq:Hta}) it follows that 
$$
J_{\langle x-1,x\rangle } = \left.\partial_s p_{x-1,x} (a_{x-1}, a_{x-1}\str, \ep{\iu s}a_{x}, \ep{-\iu s}a_{x}\str) \right\vert_{s=0}.
$$
Accordingly, the current density is given by
\begin{equation}
\label{eq:j}
j = \frac{1}{L}\sum_{x\in \ring} J_{\langle x-1,x\rangle} = \frac{1}{L} \left.\partial_s \tilde{H}_s\right\vert _{s=0},
\end{equation}
where 
$$
\tilde{H}_s = \sum_{x \in \ring} p_{x-1,x} (a_{x-1}, a_{x-1}\str, \ep{\iu s}a_{x}, \ep{-\iu s}a_{x}\str).
$$
In general, the `twist' Hamiltonian $\tilde{H}_s$ is not gauge equivalent to $H$. However, 
$$
\tilde{H}_{\frac{2 \pi}{L}} = H_{\varphi}
$$
where $\varphi$ is the gauge transformation
\begin{equation*}
\varphi_x = 2\pi\frac{x}{L},\quad x\in \ring.
\end{equation*}
Indeed, (\ref{1.2}) implies that 
$$
H_\varphi = \sum_{x \in \Pi} p_{x-1,x} (a_{x-1}, a_{x-1}\str, \ep{\iu (\varphi_x - \varphi_{x-1})}a_{x}, \ep{-\iu (\varphi_x - \varphi_{x-1})}a_{x}\str),
$$
which is equal to $\tilde{H}_{\frac{2 \pi}{L}}$ for the particular choice of $\varphi$ (note that the twisting is correct, in particular, for the edge $\langle L, 1 \rangle $).

Let now $\Omega$ be a (not necessarily unique) ground state of $H$. In this one-dimensional setting of a ring geometry, Bloch's theorem reads: There is constant $C>0$ such that
\begin{equation}\label{1dBloch}
 \vert \langle \Omega,j\Omega\rangle\vert  \leq \frac{C}{L}.
\end{equation}

To prove the claim, we expand $\langle \Omega, (\tilde H_s - H) \Omega\rangle$ to the first order. Using (\ref{eq:j}) we have,
\begin{equation*}
\langle \Omega, (\tilde H_s - H) \Omega\rangle = s L \langle \Omega,j\Omega\rangle  + \langle \Omega,R_s\Omega\rangle,
\end{equation*}
where the rest term is given by
\begin{equation*}
R_s = \frac{s^2}{2} (\left.\partial^2_{s,s}\tilde{H}_s\right\vert_{s=t})
\end{equation*}
for some $t\in(0,s)$. Hence there exists a constant $C$ such that 
\begin{equation}\label{RestEstimate}
\Vert R_s \Vert \leq s^2 C L.
\end{equation}
With this, the energy of the gauge transformed state $\Omega_\varphi = U_{\varphi}\Omega$ can be compared with the ground state energy
\begin{equation*}
0\leq \langle \Omega_\varphi, H\Omega_\varphi\rangle - \langle \Omega, H\Omega\rangle 
= \langle \Omega, (\tilde{H}_\frac{2 \pi}{L}-H)\Omega\rangle 
 =  2\pi \langle \Omega, j\Omega\rangle + \langle \Omega, R_\frac{2 \pi}{L} \Omega\rangle.
\end{equation*}
Now, the argument can be repeated with $\varphi\to-\varphi$, yielding the inequality involving $-j$. Together, one concludes that
\begin{equation*}
-\langle \Omega, R_\frac{2 \pi}{L} \Omega\rangle\vert \leq 2\pi  \langle\Omega, j\Omega\rangle  \leq  \langle \Omega, R_\frac{-2 \pi}{L} \Omega\rangle
\end{equation*}
which yields the claim~(\ref{1dBloch}) by~(\ref{RestEstimate}).

Let us first point to the strength of the argument, namely the very limited assumptions made along the way. It is valid for systems having degenerate ground states. It does not make any assumptions about the spectral gap above the ground state energy.

Another useful aspect of this variational argument built on a unitary is that it extends to thermal equilibrium states. Indeed, the free energy of a density matrix $\rho$ being
\begin{equation*}
F(\rho) = \tr(\rho H) - \beta^{-1} S(\rho),
\end{equation*}
the variation $\rho\mapsto \rho_\varphi=U_\varphi \rho U_\varphi\str $ leaves the entropy constant while the energy difference is given as above. Hence, if $\rho$ is an equilibrium state, namely a minimizer of $F$, we conclude that 
\begin{equation*}
0\leq F(\rho_{\pm\varphi}) - F(\rho) = \pm 2\pi \tr(\rho j) + \tr(\rho M_{\pm\varphi})
\end{equation*}
and the norm estimate~(\ref{RestEstimate}) again implies that $\vert \tr(\rho j)\vert \leq CL^{-1}$ for equilibrium states at finite temperature.

On the other hand, let us consider a quasi one-dimension ring of width $W$, imposing periodic boundary conditions in the transverse direction. Replacing in the discussion above the site $x$ by the full slab $[x]$ of width $W$, we obtain that
\begin{equation}\label{2dFail}
\vert \langle \Omega,J_W\Omega\rangle \vert \leq C\frac{W}{L},
\end{equation}
where $J_W$ is the current density per slab. In particular, this is too weak to prove the vanishing of the current across a full `cut' of a two-dimensional system where $W/L\to r>0$ as $L\to\infty$. This illustrates that the arguments above does not extend to higher dimensions.

We shall show in the following sections how this limitation can be overcome by assuming a spectral gap above the ground state energy, while still keeping a possible (finite) ground state degeneracy. 

\section{Absence of currents in gapped systems}
For simplicity, we phrase the result in the geometric setting of a two-dimensional torus. Let $\Lambda = (\bbZ/L\bbZ)^2$ be the discrete torus, with vertices denoted $x = (x_1,x_2)\in \Lambda$. It is equipped with a metric $d(\cdot,\cdot)$, which we take as the graph distance. The Hilbert space of the system is
\begin{equation*}
\caH = \caF_-(\bbC^{\vert \Lambda \vert}),
\end{equation*}
where $\caF_-(\bbC^{\vert \Lambda\vert})$ is the antisymmetric Fock space of $\vert \Lambda\vert$ degrees of freedom. The observables are even elements of the CAR algebra, namely linear combinations of even monomials in the fermionic creation and annihilation operators. 

We denote $Q_X = \sum_{x\in X} q_x$ the charge in a set $X$, in particular $Q_\Lambda$ is the total charge. Let $H = \sum_{X \subset \Lambda} h_X$ be a local Hamiltonian having finite range, by which we mean that 
\begin{equation*}
h_X = 0\quad \text{whenever}\quad\mathrm{diam}(X)\geq R,
\end{equation*}
and which is charge conserving,
\begin{equation*}
[h_X,Q_\Lambda] =0
\end{equation*}
for all $X$, see~(\ref{1d Charge conservation}). The spectrum of $H$ is assumed to have a gap, namely
\begin{equation}\label{gap}
\sigma(H) \subset \Sigma \cup \Sigma_+
\end{equation}
where $\mathrm{dist} (\Sigma , \Sigma_+) = \gamma>0$ uniformly in $L$. Then $P$ is the spectral projector
\begin{equation*}
P = \chi_{\Sigma}(H)
\end{equation*}
associated with $\Sigma$, which we assume to have a constant rank $p = \mathrm{rk}(P)$ for all $L$ large enough. The case we have in mind is a projection on low laying states, and we call $P$ the ground state projection.

Let $\Gamma = \{x: 0 \leq x_1\leq L/2\}$ be the half-torus, and let $\partial_\pm$ be strips of width $R$ around the boundary of $\Gamma$, i.e. $\partial_- = \{x : |x_1| \leq R \}$. We denote
$
Q = Q_\Gamma
$
the operator of charge in the half-torus. Charge conservation and the fact that $H$ has finite range implies that 
\begin{equation}\label{MB currents}
\iu [H,Q] = \iu[H_-,Q] + \iu[H_+ Q] = J_- - J_+,
\end{equation}
where $H_\pm = \sum_{X \subset \partial_\pm} h_X$ and hence  $J_\pm$ is supported in $\partial_\pm$ respectively. Since the distance between $\partial_-$ and $\partial_+$, is proportional to $L$, it is possible and will be useful
 to consider in $H_\pm$ all terms supported in wider strips $S_\pm$ that are still a distance $L$ apart, but have themselves a width of order $L$. Of course, this does not modify the operators $J_\pm$ at all.

We consider the total current through a fiducial line at $x_1=0$, namely we put $J = J_-$. By $\eqL$ we denote an equality up to $\caO(L^{-\infty})$ terms (in the topology of the norm in operator equations).

\begin{thm}
\label{Gap_Bloch}
In the setting above,
$$
\tr(P J) \eqL 0.
$$
In particular, the average current in the state $P$ vanishes in the large volume limit. 
\end{thm}

In the proof, we will use operators $K_\pm$, introduced for the present purpose in~\cite{MBIndex}, that encode charge fluctuations in the state $P$ on the boundaries $\partial_\pm$. Specifically, there exist operators $K_\pm$ such that 
\begin{enumerate}
\item $\Vert K_\pm\Vert \leq C L$,
\item $[K_\pm, A_X] = \caO(\dist(X, \partial_\pm)^{-\infty})$, where $\Vert A_X\Vert = 1$ and $\mathrm{supp}(A_X) = X$ with $\vert X\vert \leq C$,
\item $\overline Q := Q - (K_- - K_+)$ leaves the ground state space invariant, namely
\begin{equation}\label{LCF}
[\overline Q,P] = 0.
\end{equation}
\end{enumerate}
Note that (i,ii) imply that $K_\pm$ are supported in $\partial_\pm$, up to tails having a fast decay. Explicitly, let $K$ be defined by
\begin{equation}
\label{HastingsGenerator}
K:= \int_{-\infty}^{+\infty} W(t)   \ep{\iu t H} \iu[H,Q]  \ep{-\iu t H} \,dt = \widehat W(-\mathrm{ad}_H)(\iu\,\mathrm{ad}_H(Q)),
\end{equation}
with $W$ a real-valued, bounded, integrable function satisfying $W(t)=\caO(|t|^{-\infty})$ and $\widehat{W}(\omega)=-\frac{1}{\iu\omega}$ for all $|\omega|\geq \gamma$, with $\gamma$ the spectral gap. These properties imply that $[K,P]=[Q,P]$. By the Lieb-Robinson bound, we conclude that the splitting~(\ref{MB currents}) lifts to $K=K_- + K_+$.

\begin{proof}[Proof of Theorem~\ref{Gap_Bloch}]
By the support properties of $K_\pm$ and $J$, we have 
$$
J = \iu[H,K_-] + \iu[H_-, \overline{Q}] + \caO(L^{-\infty}),
$$
where we used that both $[H_-,K_+]$ and $[(H-H_-),K_-]$ are $\caO(L^{-\infty})$. With~(\ref{LCF}), we conclude that
$$
PJP = \iu[H,PK_-P] + \iu[PH_-P, \overline{Q}] + P\caO(L^{-\infty})P,
$$
and hence 
$$
\tr(PJ) = \caO(L^{-\infty}),
$$
by cyclicity of the trace.
\end{proof}

Let us make a few remarks about the result. First of all, in the present higher dimensional setting, this shows that the total current across the fiducial line $\{x_1 = 0\}$ vanishes in the large volume limit and very fast indeed, namely
\begin{equation*}
\left\vert \langle J \rangle_P \right\vert \leq \frac{C_k}{L^k}
\end{equation*}
for all $k\in\bbN$, where $\langle J \rangle_P = p^{-1} \tr(PJ)$. This should be compared with~(\ref{2dFail}). The cost of this improvement is the additional spectral gap assumption, which we have seen to be a fundamental ingredient of the proof. Secondly, in the case of $p = \mathrm{rk}(P) >1$, the vanishing may in principle be due to cancellations within the ground state space. One additional assumption ensuring that this is not the case is that of local topological order in the ground state space, namely that 
\begin{equation*}
PAP - \langle A\rangle_P P\eqL 0
\end{equation*}
for any local observable. It implies in particular that both $PH_-P$ and $PK_-P$ are proportional to $P$, since both $H_-$ and $K_-$ are sums of local terms. Hence the second line of the proof immediately gives 
\begin{equation*}
PJP \eqL 0.
\end{equation*}

We also point out that the $\caO(L^{-\infty})$ smallness of the current is truly a ground state property so that the above result does indeed not extend to thermal equilibrium states.

Operators $K_\pm$ are used in the above proof as a tool to zoom to one of the boundaries $\partial_\pm$. The technique was introduced in \cite{MBIndex} in the context of many-body index theory. In fact, Bloch's theorem is a consequence of this general theory, a connection we describe in the next section.

\subsection{Connection to a many-body index}  
We briefly recall the definition of the many-body index introduced in~\cite{MBIndex} and generalized to the degenerate case in~\cite{PRBIndex}, we refer to~\cite{TOIndex} for a complete exposition. The theory describes an index associated to a charge transported across a fiducial hyperplane.

Let $U$ be a unitary on $\caH$ that implements transport. We assume that $U$ is generated by a possibly time dependent Hamiltonian $G(s)$ for $s\in[0,1]$, which may not need to be the generator of the physical time evolution. However, $G(s)$ is assumed to be local and charge conserving, namely
\begin{equation*}
G(s) = \sum_{X\subset\Lambda} g_X(s),
\end{equation*}
where $g_X(s)$ is supported in $X$ and
\begin{equation*}
[g_X(s), Q_\Lambda] = 0,
\end{equation*}
 for all $s$. Locality is expressed in terms of the decay of the norm of $g_X(s)$ as  function of the size of $X$, for example by assuming that
\begin{equation*}
\sup_{s\in[0,1]}\sup_{x\in\Lambda}\sum_{X\ni x}\frac{\Vert g_X(s)\Vert}{\xi(\mathrm{diam}(X))} < C
\end{equation*}
uniformly in $L$, where $\xi:[0,\infty)\to (0,\infty)$ is an $L$-independent, rapidly decaying function: $\xi(r) = \caO(r^{-\infty})$. With this, $U = U(1)$ is the solution of the Schr\"odinger equation
\begin{equation*}
\iu \dot U(s) = G(s) U(s),\qquad U(0) = 1.
\end{equation*}
Its adjoint action on the observables satisfies a Lieb-Robinson bound, see for example~\cite{AmandaQL}, and hence
\begin{equation*}
U\str \caA_X U\subset\caA_X
\end{equation*}
for any set $X$, where $\caA_X$ is the set of observables supported in $X$, up to corrections whose norm vanish fast in the distance to $X$. Secondly, $U$ conserves charge in the sense that 
\begin{equation*}
U\str Q_X U - Q_X \in \caA_{\partial X},
\end{equation*}
where $\partial X = \{x: d(x,X)\leq 1\text{ and }d(x,X^c)\leq 1 \}$ is the boundary of $X$. In particular, the operator of net charge transported into the half-torus has the form
\begin{equation}\label{UCC}
U\str Q U - Q \eqL T_- - T_+,\qquad T_\pm\in\caA_{\partial_\pm}
\end{equation}
where $\partial_\pm$ are the two disjoint parts of the boundary of $\Gamma$. This of course is to be related to~(\ref{1d currents}) in the first section. Just as it was there, this specific form follows from the assumed locality and charge conservation, as
\begin{align}
U\str QU-Q &= \iu\int_0^1  U^* (s) [G(s), Q] U(s) \dd s \nonumber \\
& \eqL  \iu\int_0^1  U^* (s) [G_{-}(s) , Q] U(s) \dd s 
+ \iu\int_0^1  U^* (s) [G_{+}(s) , Q] U(s)  \dd s \label{Tminus}
\end{align}
identifies $T_\pm$. Indeed, by charge conservation the local expansion
\begin{equation*}
[G(s), Q] = \sum_{X\subset\Lambda}[g_X(s),Q]
\end{equation*}
asymptotically splits into the two contributions $X\cap \partial_-\neq\emptyset$ and $X\cap \partial_+\neq\emptyset$ where each one belong to $\caA_{\partial_-}$, respectively $\caA_{\partial_+}$, since the sets $X$ of diameter $\mathrm{diam}(X) = o(L)$ (in particular those that span both $\partial_\pm$) have vanishing contributions for large $L$. The Lieb-Robinson bound yields the claim~(\ref{UCC}). 
 
The final hypothesis of the theorem relates $P$ and $U$: the range of $P$ is asymptotically invariant under $U$, namely
\begin{equation*}
[U,P] \eqL 0.
\end{equation*}

\begin{thm}\label{thm: Index}
Under the assumptions above,
\begin{equation*}
\mathrm{dist} (\tr(P T_-), \bbZ) \eqL 0.
\end{equation*}
\end{thm}
In other words, the expected charge transport across the fiducial line (which in the present two-dimensional setting has length $L$) is an integer multiple of $1/p$ for large $L$ to almost exponential precision.

Although we shall not delve into the proof, we point out few basic ideas. First of all, the gap, together with the Lieb-Robinson bound, implies that $P$ satisfies a clustering property,
\begin{equation}\label{clustering}
P ABP - PAPBP = \min\{\vert X\vert,\vert Y \vert\}\caO(d(X,Y)^{-\infty})
\end{equation}
for any $A\in \caA_X,B\in\caA_Y$, see~\cite{HastingsClustering,BrunoClustering}. 

While $[\overline Q,P] = 0$ implies that charge fluctuations between the ground state space and its orthogonal complement vanish, which is in stark contrast with the fluctuations of the charge in the half-space $Q$, the two `charges' have the same expected transport since
\begin{equation*}
\tr(P (U\str \overline Q U - \overline Q)_- ) \eqL \tr(PT_-) - \tr (P(U\str K_- U - K_-)) \eqL \tr(PT_-)
\end{equation*}
by the support property of $K_-$ in the first equality and cyclicity of the trace with $[U,P]\eqL0$ in the second one. The proof proceeds by computing the full counting statistics~\cite{KlichFCS,avron2008fredholm} of $\overline Q$ though the fiducial line $\partial_-$, which is associated with the operator $Z_-(\phi)$ defined by the factorization
\begin{equation*}
U\str \ep{\iu\phi \overline Q} U \ep{-\iu\phi \overline Q} = Z_-(\phi) Z_+(\phi).
\end{equation*}
This equality further allows us to point to the use of clustering in the proof. Indeed, the unitary operator $Z(\phi)$ on left hand side leaves the ground state space invariant by assumption on $U$ and construction of $\overline Q$. But clustering implies that
\begin{equation*}
PZP \eqL PZ_-PZ_+P
\end{equation*}
so that 
\begin{equation*}
1\geq \Vert PZ_-P\Vert \geq \Vert PZP\Vert\eqL 1
\end{equation*}
which proves that $[Z_-(\phi),P]\eqL0$, namely that the ground state space is an invariant space of $Z_-(\phi)$, too.

With the index in hand, Bloch's theorem is an elementary corollary of Theorem~\ref{thm: Index}. For any $t\in[0,1]$, we pick
\begin{equation*}
U = \ep{-\iu t H}
\end{equation*}
which is local, conserves charge and commutes with $P$. Now, the operators on the r.h.s.~of~(\ref{MB currents}) are naturally identified with the currents across the  lines $\partial_\pm$, see~(\ref{1d currents}).

The charge transported in the interval $[0,t]$ is explicit, see~(\ref{Tminus}),
\begin{equation*}
\tr (P T_-(t)) 
=\iu\int_0^t  \tr (P U^* (s) [H_- , Q] U(s)) \dd s
= t\tr (P J)
\end{equation*}
since $U(s) P U(s)\str = P$ for any $s\in[0,1]$. All assumptions of Theorem~\ref{thm: Index} apply, so that
\begin{equation*}
\mathrm{dist} (t \tr (P J),\bbZ) \eqL 0,
\end{equation*}
 and since this is valid for all $t\in[0,1]$, we conclude that
\begin{equation*}
 \langle J\rangle_P \eqL 0,
\end{equation*}
which is Bloch's theorem again.

\section{Currents in mesoscopic rings}

A recurring question associated with Bloch's theorem is its apparent contradiction with the existence of superconducting currents. A short answer is that persistent currents in superconducting rings is a mesoscopic phenomenon. We are, however, not aware of a concrete microscopic model to demonstrate this point. On the other hand, the related persistent currents in mesoscopic metallic rings \cite{Buttiker1983} is modelled by a free Laplacian on a ring pierced by a magnetic flux. In this model, the current can be explicitly calculated. Apart from showing that it indeed vanishes in the large volume limit, the example also illustrates that the gap condition in Theorem~\ref{Gap_Bloch} is necessary: As we will see, the model is gapless and the current is of order $L^{-1}$ instead of $L^{-\infty}$ that would be guaranteed by the theorem, had the model have a gap.

We describe a lattice version of the model \cite{AFG}. A single particle Hamiltonian associated to an electron hopping on a ring $\ring=\bbZ/L\bbZ$ threaded by a flux $\phi \in [0, 2 \pi)$ is
$$
H = - \ep{\iu \phi/L} T - \ep{-\iu \phi/L} T^*,
$$
where
\begin{equation*}
(T \psi)(x) = \psi(x-1),
\end{equation*}
and acts on the Hilbert space $l^2(\ring)$. The normalized eigenstates of $H$ are given by 
$$
\psi_k(x) = \frac{1}{\sqrt L} \ep{2\pi \iu x k/L},\qquad k\in\bbZ/L\bbZ,
$$
with eigenvalues
\begin{equation*}
H \psi_k = -2 \cos \left( \frac{\phi - 2 \pi k}{L} \right) \psi_k.
\end{equation*}
The Hamiltonian does not have any gap in the spectrum that remains open in the large volume limit.

We denote by $\{| x \rangle ,\, x=0, \dots , L-1\}$ the standard position eigenbasis and note that $T = \sum_{x\in\Pi}\vert x+1\rangle\langle x\vert $. The charge in the interval $[a,b]$ is given by $Q_{[a,b]} = \sum_{x=a}^b |x \rangle \langle x|$, namely it is the multiplication operator by the indicator function of the interval $[a,b]$. We have 
$$
\iu[H, Q_{[a,b]}] = J_{\langle a-1,a \rangle} - J_{\langle b,b+1 \rangle},
$$
where
\begin{equation*}
J_{\langle x-1,x \rangle} =  \iu \ep {\iu \phi/L} |x \rangle \langle x-1| - \iu \ep{-\iu \phi/L} | x-1 \rangle \langle x|.
\end{equation*}
The current per edge is 
$$
j = \frac{1}{L} \sum_{x=0}^{L-1} J_{\langle x-1, x \rangle} = \frac{\iu}{L}\ep{\iu \phi/L} T - \frac{\iu}{L} \ep{-\iu \phi/L} T^*.
$$
Note that $j = -\partial_\phi H$, and hence we get
$$
\langle \psi_k , j \psi_k \rangle = \frac{2}{L} \sin \left(  \frac{2 \pi k-\phi}{L} \right).
$$
 By translation invariance, the expectation values of $j$ and $j_{\langle x-1,x \rangle}$ are the same in any eigenstate of $H$.

The ground state of $N$ non-interacting electrons is given by a Fermi projection $P_F$ on the $N$ lowest energy levels of $H$. For $\phi \in (0, \pi)$ and $N = 2m+1$ (the even case being similar), this corresponds to the eigenvectors $\psi_k$ with $k$ in the interval $[-m,m]$. The current expectation value is then 
$$
\tr ( P_F j) = \sum_{k=-m}^m \frac{2}{L} \sin \left( \frac{ 2 \pi k -\phi}{L} \right).
$$
The sum can be explicitly computed, and in the limit $L \to \infty$ with $N/L \to \rho <1$ we get 
$$
\tr ( P_F j) = -\frac{1}{L} \frac{2 \phi}{\pi} \sin(\pi \rho) + \caO(L^{-2}).
$$
We conclude that the current is indeed of order $1/L$. It vanishes in the large volume limit but only polynomially in $L$ showing that the gap assumption in  Theorem~\ref{Gap_Bloch} is necessary. Note that for $\phi \neq 0$, the time-reversal invariance is broken. In the time-reversal invariant situation $\phi = 0$, the current vanishes identically. 

\section*{Acknowledgements}

\noindent The results presented in this note owe very much to our collaboration with W.~De Roeck. We thank H.~Watanabe for an email correspondence related to this work. The work of S.B. was supported by NSERC of Canada. M.F. was supported in part by the NSF under grant DMS-1907435.

\end{document}